\documentclass[12pt]{article}
\usepackage{graphicx} 
\usepackage{amsmath} 
\usepackage{amssymb} 
\usepackage{cite}      
\usepackage{rotating}  
\usepackage{booktabs}  
\usepackage{tikz} 
\usepackage{algorithm}
\usepackage{algorithmic}

\usepackage{booktabs}  
\usepackage{graphicx}
\usepackage{amsmath}
\usepackage{amssymb}  
\usepackage{caption}
\usepackage{algorithmic}
\usepackage{framed}
\usepackage{rotating}
\usepackage{tabularx}
\usepackage{longtable}
\usepackage{wrapfig} 
\usepackage{url}      
\usepackage{hyperref} 
\usepackage{subcaption}
\usepackage{color}
\usepackage{doc}  
\usepackage{multirow}
\usepackage{dcolumn}
\usepackage{cite}
\usepackage{listings}
\usepackage{tikz} 
\usepackage{amsthm}

\newtheorem{theorem}{Theorem}
\newtheorem{lemma}{Lemma}

\definecolor{pale}{rgb}{1,0.85,0.85}
\definecolor{pale2}{rgb}{1,0.85,0.15}

\makeatletter
\def\Ddots{\mathinner{\mkern1mu\raise\p@
  \vbox{\kern7\p@\hbox{.}}\mkern2mu
  \raise4\p@\hbox{.}\mkern2mu\raise7\p@\hbox{.}\mkern1mu}}
  \makeatother

\title{The Golay Code Outperforms the Extended Golay Code Under
Hard-Decision Decoding}
\author{Jon Hamkins%
\thanks{This work was done as a private effort and not in the author's
capacity as an employee of the Jet Propulsion Laboratory, California
Institute of Technology.}}

\begin{document}
\maketitle

\bigskip

\centerline{\it Submitted to: IEEE Transactions on Information Theory}
\centerline{\it Keywords: Error correction codes, performance analysis}

\bigskip
\bigskip

\begin{abstract}
We show that the binary Golay code is slightly more power efficient than
the extended binary Golay code under maximum-likelihood (ML),
hard-decision decoding.  In fact, if a codeword from the extended code
is transmitted, one cannot achieve a higher probability of correct
decoding than by simply ignoring the 24\textsuperscript{th} symbol and
using an ML decoder for the non-extended code on the first 23 symbols.
This is so, despite the fact that using that last symbol would allow one
to sometimes correct error patterns with weight four.  To our knowledge
the worse performance of the extended Golay code has not been previously
noted, but it is noteworthy considering that it is the extended version
of the code that has been preferred in many deployments.  
\end{abstract}

\section{Introduction}
The many interesting properties of the Golay codes are well-studied
\cite{Golay1949, Lin1983}, and the codes have been deployed in several
applications.  The extended binary Golay code was used for NASA's
Voyager mission during its encounters at Jupiter and Saturn
\cite{Laeser1986}.  It was also used to protect the data handling
capabilities of NASA's Magellan mission to Venus \cite{Butrica1996}, and
the nonimaging science experiments of the Galileo mission
\cite{Butrica1996}.  Outside of NASA, the extended Golay
code has been used as part of the Automatic Link Establishment protocol
ITU-R F.1110 \cite{ITU-F.1110-2.1997}, it has been used in paging
protocols \cite{ITU-900-2-1990}, and it remains a standard for telemetry
\cite{IRIG2015}.

One property that hasn't been reported in the literature, to the best of
our knowledge, is that under maximum-likelihood (ML) decoding with
hard-decisions, the extended binay Golay code, ${\cal G}_{24}$ has
slightly {\em worse} power efficiency than the binary Golay code, ${\cal
G}_{23}$.  This paper demonstrates this fact.

\section{Performance of ML decoding with hard-decisions}
\label{sec.Golayperf-hard}
In the following, we assume a binary symmetric channel, corresponding to
a receiver that makes hard decisions.  For the extended Golay code, the
literature often describes an incomplete decoder capable of correcting
three errors and detecting four errors.  Such a decoder does not
minimize the probability of codeword error, because it makes no attempt
to guess at the correct codeword when four errors are detected.  Since
we desire to compare best-achievable error rate performance, we focus on
codeword-error-rate-minimizing decoders, i.e., complete ML decoders that
always output a closest codeword.  The ML decoder offers no error
detection, and we do not address error detection in this paper.

\subsection{Binary Golay Code, ${\cal G}_{23}$}
A complete ML decoder for the (23,12,7) Golay code will produce the
correct codeword at its output whenever a hard-decision channel makes
three or fewer errors in the 23 symbols.  Thus, the codeword error
rate (CWER) is
\begin{equation}
  \mbox{CWER} = 1 - \sum_{i=0}^{3} \binom{23}{i} p^i (1-p)^{23-i}
  \label{eq.Golay-CWER}
\end{equation}
where $p$ is the channel symbol error probability.  This performance is
shown in Figure~\ref{fig.Golay-perf} for a binary-input, additive white
Gaussian noise (AWGN) channel with hard decision error probability
\begin{equation}
  p = Q\left(\sqrt{\frac{2E_s}{N_0}}\right)
  = Q\left(\sqrt{\frac{2RE_b}{N_0}}\right)
  \label{eq.BSC-p}
\end{equation}
where $R=12/23$ is the code rate.

\begin{figure}[p]
  \centerline{\includegraphics[height=0.9\textheight]{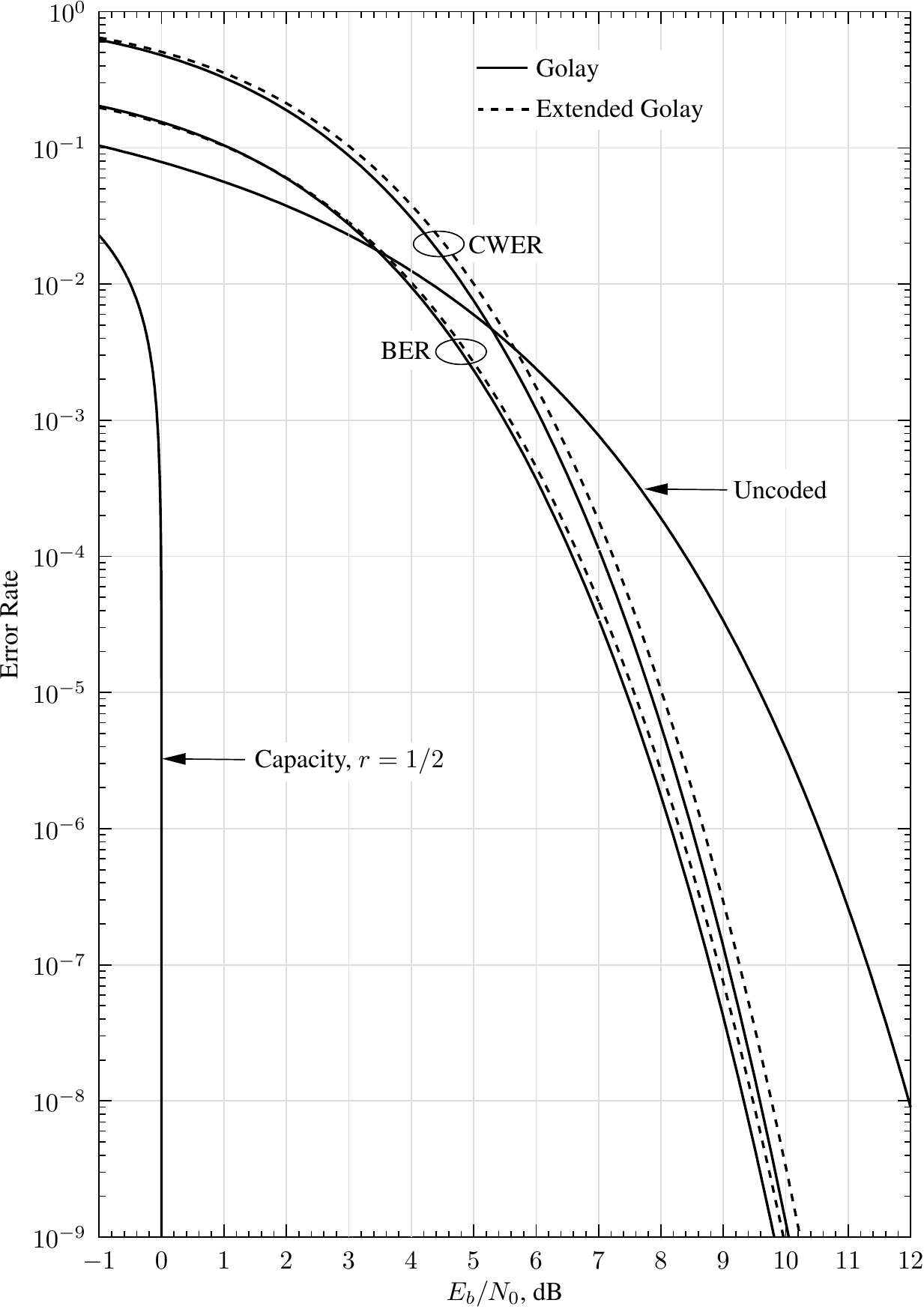}}
  \caption{Performance of the binary Golay code and extended binary
Golay code under hard-decision ML decoding, compared to capacity and
uncoded transmission.}
\label{fig.Golay-perf}
\end{figure}

The bit error-rate (BER) performance is also shown in
Figure~\ref{fig.Golay-perf}.  The BER is difficult to express
analytically when the signal-to-noise ratio (SNR) is low, but the
analysis becomes feasible at moderate to high SNR, because with high
probability each codeword error results in seven code symbol errors.
Each symbol (whether a systematic information-bearing symbol or a
parity symbol) has 7/23 chance of being in error.  Thus, at moderate
and high SNR, the bit error rate is given by
\begin{equation}
\mbox{BER} \approx \frac{7}{23} \mbox{CWER} 
\label{eq.Golay-BER}
\end{equation}
which is tight (within about 0.1 dB) when $E_b/N_0 > 1$ dB.  

\subsection{Extended Binary Golay Code, ${\cal G}_{24}$}
Determining the performance of a complete ML decoder for ${\cal G}_{24}$
is more involved because, unlike ${\cal G}_{23}$, the closest codeword
to a received vector may not be unique.  Determining decoder performance
requires knowing how many codewords can be tied in this way, and how
many bit errors are produced if the decoder guesses the wrong one.  To
help us, we start with two lemmas.

\begin{lemma} \cite{Berlekamp1972,McEliece2002}
  There is a unique codeword of ${\cal G}_{24}$ of weight eight which has ones
  in any five given positions.
  \label{lem.1-5}
\end{lemma}

\begin{lemma}
  There are five distinct codewords of ${\cal G}_{24}$ of weight
  eight which have ones in any four given positions.
\label{lem.5-4}
\end{lemma}
\begin{proof}
Let ${\bf y} = (y_0, \ldots, y_{23})$ be a vector with ones in four
given positions.  Let ${\bf y}' = {\bf y} + {\bf u}_i$, where $i$ is one
of the 20 indices for which $y_i=0$, and where ${\bf u}_i$ is a vector
with a one in the $i$\textsuperscript{th} position.  By
Lemma~\ref{lem.1-5}, there is a unique codeword of ${\cal G}_{24}$ of weight
eight which has ones in the same five positions as ${\bf y}'$.  Repeat
this argument for each of the 20 values of $i$ for which $y_i=0$.  This
yields a list of 20 codewords.  Each of these codewords occurs in the
list four times, corresponding to the four positions that ones have been
added to ${\bf y}$.  Thus, there are 20/4 = 5 distinct codewords of
weight eight which have ones in the same four positions as ${\bf y}$.
\end{proof}

We are now ready to state the distance properties we need to evaluate
the performance of a complete ML decoder for ${\cal G}_{24}$.
\begin{theorem}
\label{thm.Golay}
  For any vector ${\bf y} \in \{0,1\}^{24}$, either
  \begin{itemize}
    \item There is a unique codeword ${\bf c} \in {\cal G}_{24}$ with 
      $d({\bf c},{\bf y})\le 3$, or
    \item There are six distinct codewords ${\bf c}\in {\cal G}_{24}$ with 
      $d({\bf c},{\bf y})=4$.
  \end{itemize}
\end{theorem}
\begin{proof}
If there is a codeword ${\bf c}$ with $d({\bf c},{\bf y})\le 3$, 
then it must be unique, for it there were two codewords ${\bf c}^{(1)}$
and ${\bf c}^{(2)}$ of ${\cal G}_{24}$ each within distance three of
${\bf y}$, then
\begin{equation}
  d({\bf c}^{(1)},{\bf c}^{(2)}) \le d({\bf c}^{(1)},{\bf y}) + d({\bf
  y},{\bf c}^{(2)}) \le 3 + 3 = 6.
\end{equation}
Since ${\cal G}_{24}$ has minimum distance eight, it must be that
${\bf c}^{(1)}={\bf c}^{(2)}$.

Now suppose there is no codeword in ${\cal G}_{24}$ within distance
three of ${\bf y}$.  There is a codeword of ${\cal G}_{23}$ at distance
at most three of $(y_0, \ldots y_{22})$ and adding a parity to this
codeword produces a codeword of ${\cal G}_{24}$ at distance at most four
from ${\bf y}$.  Since the distance is not three or less, it must be
four.

We now determine the number of codewords of ${\cal G}_{24}$ at distance
four from ${\bf y}$.  Since the code is linear, we lose no generality by
assuming that one of the nearest codewords is the all-zero codeword, and
thus, $w({\bf y})=4$.  By Lemma~\ref{lem.5-4}, there are five distinct
codewords of weight eight and distance four from ${\bf y}$.  Thus,
together with the all-zero codeword, there are six codewords of ${\cal
G}_{24}$ which are distance four from ${\bf y}$.
\end{proof}

To summarize, the ML decoder for ${\cal G}_{24}$ will find a unique
codeword if there is a codeword within distance three from the received
vector, and otherwise it will output one of the six codewords at
distance four (and it will have a 1/6 chance of being correct in that
case).  Thus, the codeword error rate is given by
\begin{equation}
  \mbox{CWER}  = 1 
  - \frac{1}{6}\binom{24}{4}p^4(1-p)^{20}
  - \sum_{i=0}^{3} \binom{24}{i} p^i (1-p)^{24-i}
  \label{eq.extended-Golay-CWER}
\end{equation}
This performance is shown in Figure~\ref{fig.Golay-perf}.  At
moderate to high SNR, when a codeword error is made then with high
probability it results in exactly eight code symbol errors.  So on
average, the code symbols have 8/24 = 1/3 chance of being in error.
When the decoder detects four errors, if it randomly selects
from among the six codewords at distance four, this 1/3 average
applies equally to the systematic and parity bits, in which case at
moderate and high SNR the bit error rate would be
\begin{equation}
  \mbox{BER}  \approx \frac{1}{3} \times\mbox{CWER}
  \label{eq.extended-Golay-BER}
\end{equation}
But the decoder need not randomly select from among the six codewords
at distance four.  Instead, the decoder may select a codeword at
distance four whose systematic bits agree with the systematic bits of
the received vector in the most number of positions.  This does not
affect the CWER, but when a codeword error is made, there are fewer
systematic bit errors than parity symbol errors, on average.
%
%
In fact, a simulation indicates that only about 3.1 of the twelve systematic
bits are in error per codeword in error (instead of four in twelve for the
decoder which randomly selects the codeword at disance four), or
\begin{equation}
  \mbox{BER}  \approx 0.26\times \mbox{CWER}
  \label{eq.extended-Golay-BER2}
\end{equation}
The performance of this decoder is shown in Figure~\ref{fig.Golay-perf}.
At BER$=10^{-6}$, ${\cal G}_{24}$ has a coding gain of 2.1 dB, and a gap
of 8.4 dB to the capacity of rate 1/2 coding on an uncontrained-input
channel.

The analysis above assumes the decoder is required to output a codeword.
If all one cares about is BER, and not producing a valid decoded
codeword, the BER can be improved further.  When four errors are
detected, one can simply output the received systematic bits exactly as
they were received from the channel.  This is unlikely to be fully
correct, but on average these bits contain only half of the four errors,
or 2 bit errors per codeword.  This compares favorably to the decoder
above, which when faced with 4 channel errors decodes to the correct
codeword 1/6 of the time and the other 5/6 of the time produces on
average 3.1 bit errors, or about 2.6 bit errors per codeword.

\subsection{Comparison of ${\cal G}_{23}$ and ${\cal G}_{24}$
Performance}
\label{sec.G23-vs-G24}
One might expect that, because the minimum distance of ${\cal
G}_{24}$ is higher than that of ${\cal G}_{23}$ (8 vs.\ 7), it is a
better code that is more efficient at correcting errors.  Remarkably,
when the channel makes hard decisions, this is not the case, even at
high SNR!  Figure~\ref{fig.Golay-perf} shows that ${\cal G}_{23}$
performs better than ${\cal G}_{24}$ by about 0.2 dB, but it is
helpful to explain why.  The reason is that transmitting the parity
bit uses slightly more energy than is saved by being able to correct
1/6 of the weight-four error patterns.  

To see this, suppose codewords of ${\cal G}_{24}$ are transmitted on a
binary symmetric channel with cross-over probability $p$.  We compare
two decoders:
\begin{itemize}
  \item Decoder $D_{23}$ is a complete ML decoder for ${\cal G}_{23}$,
    and ignores the 24\textsuperscript{th} symbol.
  \item Decoder $D_{24}$ is a complete ML decoder for ${\cal G}_{24}$.
\end{itemize}
Which decoder has a better CWER?  We can answer this by comparing two
quantities.
\begin{enumerate}
  \item {\em Error patterns that decoder $D_{24}$ corrects that decoder $D_{23}$
    does not.}\\
If the channel makes four errors in the first 23 symbols and the
24\textsuperscript{th} symbol is received correctly, then 
$D_{23}$ will decode in error, and 
with probability 1/6 $D_{24}$ will decode correctly.  Thus, the probability
$D_{24}$ is correct and $D_{23}$ is not, is
\begin{equation}
  \frac{1}{6}\binom{23}{4} p^4(1-p)^{20}
\end{equation}
\item {\em Error patterns that decoder $D_{23}$ corrects that decoder $D_{24}$
  does not.}\\
If the channel makes three errors in the first 23 symbols and the
24\textsuperscript{th} symbol is also received in error, decoder
$D_{23}$ will find the correct answer, and with probability 5/6
decoder $D_{24}$ will not find the correct answer.  Thus, the
probability $D_{23}$ is correct and $D_{24}$ is not, is
\begin{equation}
  \frac{5}{6}\binom{23}{3} p^4(1-p)^{20}
\end{equation}
\end{enumerate}
In all other situations, either both decoders find the correct
codeword, or both decoders produce a codeword error.  Since
\begin{equation}
  5\binom{23}{3} = \binom{23}{4} = 8855
\end{equation}
it follows that for any given $p$, decoder $D_{23}$ and $D_{24}$ have
the identical CWER!  Thus, if a codeword from ${\cal G}_{24}$ is
transmitted, we cannot do better than simply ignoring the
24\textsuperscript{th} symbol and using the complete ML decoder for
${\cal G}_{23}$ on the first 23 symbols.  This is so, despite the fact
that using that last symbol would allow us to sometimes correct error
patterns with weight four---that advantage is exactly balanced by
the chance that the last symbol will be received in error and prevent
proper decoding.

The story is more nuanced if we compare the BER.  By comparing
(\ref{eq.Golay-BER}) to (\ref{eq.extended-Golay-BER}), we see that for
any given $p$, the random-ML-codeword decoder for ${\cal G}_{24}$ is
{\em worse} than the ML decoder for ${\cal G}_{23}$ by a factor of
$(1/3)/(7/23) = 23/21$.  On the other hand, for a given value of $p$,
the carefully designed BER-minimizing decoder for $D_{24}$ discussed
above (see (\ref{eq.extended-Golay-BER2})) has a lower BER than that of
decoder $D_{23}$, since $0.26 < 7/23$.

This analysis addresses the question of what to do if the
24\textsuperscript{th} symbol has been transmitted.  Since the CWER is
the same whether we make use of it or not, we are even better off by not
transmitting it at all.  This allows a savings in energy of of
$10\log_{10} 24/23 \approx 0.18$ dB.  This is why the hard-decision CWER
curves in Figure~\ref{fig.Golay-perf} are separated by exactly this
amount at all SNRs.  The BER curves, as expected, are slightly closer
together, but ${\cal G}_{23}$ is still seen to be better than ${\cal
G}_{24}$ by about 0.13 dB.

We take care to note that ${\cal G}_{23}$ outperforms ${\cal G}_{24}$
only on a {\em hard-decision} channel.  Under ML soft-decision decoding,
${\cal G}_{24}$ outperforms ${\cal G}_{23}$, which can be verified by
simulation or, at high SNR, by comparing the union bound expressions for
the two codes.

\bibliographystyle{IEEEtran}
\bibliography{master}
\end{document}